\documentclass[12pt, reqno]{amsart}
\usepackage{latexsym, enumerate, imakeidx,mathrsfs, color,amscd,amsmath,amsfonts,amssymb,geometry,color,amsthm, float}
\usepackage{booktabs,array}
\usepackage[bookmarksnumbered, colorlinks, plainpages]{hyperref}
\hypersetup{colorlinks=true,allcolors=black}
\newtheorem{theorem}{Theorem}[section]
\newtheorem{lemma}[theorem]{Lemma}
\newtheorem{proposition}[theorem]{Proposition}
\newtheorem{corollary}[theorem]{Corollary}
\theoremstyle{definition}
\newtheorem{definition}[theorem]{Definition}
\newtheorem{example}[theorem]{Example}

\theoremstyle{remark}

\numberwithin{equation}{section}

\DeclareMathOperator{\Tr}{Tr}
\DeclareMathOperator{\rank}{rank}
\DeclareMathOperator{\ran}{ran}
\DeclareMathOperator{\spanop}{span}
\DeclareMathOperator{\diag}{diag}
\DeclareMathOperator{\supp}{supp}

\newcommand{\C}{\mathbb C}
\newcommand{\R}{\mathbb R}
\newcommand{\B}{\mathcal B}
\newcommand{\D}{\mathcal D}
\newcommand{\Herm}{\operatorname{Herm}}
\newcommand{\Pone}{\mathcal P_1}
\newcommand{\RankTwo}{\mathcal R_2}
\newcommand{\Sec}{\mathcal S_2^0}
\newcommand{\pr}{\operatorname{pr}}
\newcommand{\Cr}{\operatorname{Cr}}
\newcommand{\norm}[1]{\left\lVert #1\right\rVert}
\newcommand{\abs}[1]{\left|#1\right|}
\newcommand{\inner}[2]{\left\langle #1,#2\right\rangle}

\newcommand{\rtensor}{\mathbin{\overline{\otimes}}}

\begin{document}

\title[Frame Phase Retrievability and Quantum Channels]
{FRAME PHASE RETRIEVABILITY AND STATE DISTINGUISHABILITY OF QUANTUM CHANNELS}

\author[D. Han]{Deguang Han}
\address{Department of Mathematics, University of Central Florida\\
Orlando, FL 32816, USA}
\email{deguang.han@ucf.edu}

\author[K. Liu]{Kai Liu$^{*}$}
\thanks{$^{*}$Corresponding author: Kai Liu.}
\address{Department of Mathematics, University of Central Florida\\
Orlando, FL 32816, USA}
\email{kai.liu@ucf.edu}


\makeatletter
\renewcommand{\subjclassname}{\textup{2020} Mathematics Subject Classification}
\makeatother

\begin{abstract}
This survey introduces the role of frame phase retrieval in pure-state identification and information preservation by quantum channels.  For unit vectors, the lift $x\mapsto xx^*$ identifies vectors that differ only by a global phase with the same rank-one quantum state, converts frame intensities into linear functionals of the lifted state, and makes the adjoint pullback of output observables an operator-valued measurement on the input system.  From this viewpoint, a channel is phase retrievable exactly when it is injective on pure states.  We develop this correspondence through Choi-rank-two linear combinations of Kraus operators, higher-rank relative spectra, structural obstructions, and constructions with prescribed Choi rank.  We distinguish injective identification from full tomography, perfect one-shot discrimination, zero-error classical communication, and exact quantum correction.  Twirling channels then provide a structured setting in which commutants, irreducible dimensions, multiplicities, orbit-frame orthogonality, coding indices, and phase-retrievable subspaces can be read from group representations.  
\end{abstract}

\maketitle

\section{Introduction}
\label{sec:introduction}

Phase retrieval asks whether a vector can be recovered, up to a scalar of unit modulus, from quadratic measurements.  In quantum mechanics, this ambiguity is not a defect: unit vectors $x$ and $e^{\mathrm i\theta}x$ represent the same pure state $xx^*$.  The following connections
\begin{equation}
 x\sim e^{\mathrm i\theta}x
 \quad\longleftrightarrow\quad xx^*,
 \qquad
 |\inner{x}{f_j}|^2=\Tr(xx^*f_jf_j^*)
 \label{eq:basic-lift}
\end{equation}
therefore turn a phaseless frame measurement into a Born expectation of a rank-one positive operator.  If a channel $\Phi:\B(\mathcal H_A)\to\B(\mathcal H_B)$ acts between the preparation and the measurement, then
\begin{equation}
\begin{split}
 x\text{ modulo phase}
 &\longmapsto xx^*
 \xrightarrow{\ \Phi\ }\Phi(xx^*)
 \xrightarrow{\ F_j\ }
 \Tr\bigl(\Phi(xx^*)F_j\bigr)  \\
 &=\Tr\bigl(xx^*\Phi^*(F_j)\bigr).
\end{split}
\label{eq:information-flow}
\end{equation}
Thus observing the output is equivalent to measuring the pulled-back operators $\Phi^*(F_j)$ at the input.  This is the central bridge of this survey: frame theory supplies the geometry of the measurement, while the channel determines which input observables remain accessible.

The question organizing the discussion is consequently precise: \emph{When does a quantum channel retain enough information to distinguish all pure input states, and how can frame geometry characterize or construct this property?} The framework in \cite{LiuHan2025} introduced phase-retrievable channels and developed Kraus-operator and relative-spectrum criteria, together with frame-based constructions.  The subsequent twirling channel theory in \cite{HanLiu2024} uses representation theory to compute phase-retrievable subspaces and several zero-error quantities.  The operator-valued phase-retrieval results in \cite{HanJuste2019} supply selected background for passing from scalar frame coefficients to general quadratic measurements.

Throughout, ``distinguishability'' means injective identification unless stated otherwise.  Perfect single-shot discrimination, full informational completeness, and exact correction are different tasks.
Section~\ref{sec:distinction} compares them and presents their operational relations \cite{Helstrom1976,Watrous2018}.

The survey is organized as follows.  Section~\ref{sec:measurement-geometry} develops concepts of the lift, positive operator measurements, POVMs, and channel pullbacks.  Section~\ref{sec:criteria-design} gives pure-state criteria, stability, a two-Kraus linear combination criterion, a higher-rank relative-spectrum test, and related channel designs.  Section~\ref{sec:distinction} compares identification, zero-error discrimination, and correction. Section~\ref{sec:twirling} illustrates the general theory through the representation-theoretic structure of twirling channels.  The final section draws on related work in the literature to introduce further connections between frame theory and quantum information and to outline promising directions and open questions for future research.

%%%%%%%%%%%%%%%%%%%%%%%%%%%%%%%%%%%%%%%%%%%%%%%%%%%%%%%%%%%%%%%%%%%%%%%%%%%%%%%%%%%%%%%%%%%%%%%%%%%%%%%%%%%%%%%%%
%%%%%%%%%%%%%%%%%%%%%%%%%%%%%%%%%%%%%%%%%%%%%%%%%%%%%%%%%%%%%%%%%%%%%%%%%%%%%%%%%%%%%%%%%%%%%%%%%%%%%%%%%%%%%%%%%
%%%%%%%%%%%%%%%%%%%%%%%%%%%%%%%%%%%%%%%%%%%%%%%%%%%%%%%%%%%%%%%%%%%%%%%%%%%%%%%%%%%%%%%%%%%%%%%%%%%%%%%%%%%%%%%%%
%%%%%%%%%%%%%%%%%%%%%%%%%%%%%%%%%%%%%%%%%%%%%%%%%%%%%%%%%%%%%%%%%%%%%%%%%%%%%%%%%%%%%%%%%%%%%%%%%%%%%%%%%%%%%%%%%
%%%%%%%%%%%%%%%%%%%%%%%%%%%%%%%%%%%%%%%%%%%%%%%%%%%%%%%%%%%%%%%%%%%%%%%%%%%%%%%%%%%%%%%%%%%%%%%%%%%%%%%%%%%%%%%%%

\section{Frames, operator measurements, and channel pullbacks}
\label{sec:measurement-geometry}

\subsection{Frames and the rank-one lift}

All Hilbert spaces are finite dimensional and complex unless the real field is explicitly indicated.  Inner products are linear in the first variable. For $x,y\in\mathcal H$, the notation $x\rtensor y=xy^*$ means the
rank-one map $z\mapsto\inner{z}{y}x$.  We write
 $\D(\mathcal H)=\{\rho=\rho^*:\rho\ge0,\ \Tr\rho=1\},
 \;
 \Pone(\mathcal H)=\{xx^*:\norm{x}=1\}.$
The self-adjoint space $\Herm(\mathcal H)$ has the Hilbert--Schmidt inner product $\inner{X}{Y}_{\rm HS}=\Tr(XY)$ and norm $\norm{X}_2$.

A family $\mathcal F=\{f_j\}_{j=1}^N\subset\mathcal H$ is a frame if
\begin{equation*}
 A\norm{x}^2\le \sum_{j=1}^N|\inner{x}{f_j}|^2\le B\norm{x}^2
 \qquad(x\in\mathcal H)
\end{equation*}
for some $0<A\le B<\infty$.  It is Parseval when $\sum_jf_jf_j^*=I$.  The positive invertible operator $S_{\mathcal F}=\sum_{j=1}^N f_jf_j^*$ is the \emph{frame operator}; replacing $f_j$ by $S_{\mathcal F}^{-1/2}f_j$ is the normalization of the frame.  The family is phase retrievable if its intensity measurements determine every vector uniquely up to a global phase. More precisely, if
    $|\langle x,f_i\rangle|^2 = |\langle y,f_i\rangle|^2 \; \text{for all } i,$ then \(y=e^{\mathrm{i}\theta}x\) for some $\theta\in\mathbb{R}$, or equivalently, $xx^*=yy^*$. Over a real Hilbert space, this condition becomes $y=\pm x$.  For standard background on frames, frame operators, and canonical Parseval frames, see \cite[Chapters~1--2]{Christensen2016}. Define the lifted analysis map
\begin{equation}
 \mathcal A_{\mathcal F}:\Herm(\mathcal H)\longrightarrow\R^N,
 \qquad
 \mathcal A_{\mathcal F}(X)=\bigl(\Tr(Xf_jf_j^*)\bigr)_{j=1}^N.
 \label{eq:lifted-analysis}
\end{equation}
For a subspace $\mathcal M\subseteq\mathcal H$, put
\begin{equation}
 \begin{aligned}
 \RankTwo(\mathcal M)
   &=\{X\in\Herm(\mathcal M):\rank X\le2\},\\
 \Sec(\mathcal M)
   &=\{X\in\Herm(\mathcal M):\Tr X=0,\ \rank X\le2\}.
 \end{aligned}
 \label{eq:secant-cone}
\end{equation}
The first set is the relevant lifted kernel variety for unconstrained vectors. The second is its normalized pure-state secant cone.

The following standard lifted-kernel criterion is the finite-dimensional
formulation used in frame phase retrieval
\cite{BalanCasazzaEdidin2006,BandeiraEtAl2014}.

\begin{theorem}[Lifted kernel criterion]
\label{thm:lifted-kernel}
Consider self-adjoint operators $M_1,\ldots,M_N$ on $\mathcal H$, and let $\mathcal A:\Herm(\mathcal H)\to\R^N$ be defined by $\mathcal A(X)=(\Tr(XM_j))_{j=1}^N$.  The measurements $\{\Tr(xx^*M_j)\}$ determine $xx^*$ for every $x\in\mathcal H$ if and only if
\begin{equation}
 \ker\mathcal A\cap\{xx^*-yy^*:x,y\in\mathcal H\}=\{0\}.
 \label{eq:exact-secant-kernel}
\end{equation}
If the $M_j$ are positive and their sum is positive definite, this is equivalent to
\begin{equation}
 \ker\mathcal A\cap\RankTwo(\mathcal H)=\{0\}.
 \label{eq:rank-two-kernel}
\end{equation}
In particular, \eqref{eq:rank-two-kernel} characterizes phase retrieval by a vector frame with $M_j=f_jf_j^*$.
\end{theorem}

\begin{proof}[Proof sketch]
Equality of the measurement lists for $x$ and $y$ is exactly the assertion $xx^*-yy^*\in\ker\mathcal A$, proving the first statement.  Every such difference is self-adjoint and has rank at most two.  Conversely, positivity and the positive definiteness of $\sum_jM_j$ exclude a nonzero positive or negative semidefinite kernel element.  Hence every nonzero element of $\ker\mathcal A\cap\RankTwo(\mathcal H)$ has one positive and one negative eigenvalue and can be written as $xx^*-yy^*$.  This gives the second formulation.  See \cite{BalanCasazzaEdidin2006,BandeiraEtAl2014} for the classical frame criterion and its phase-retrieval interpretation.
\end{proof}

If $S$ is invertible, then $\{f_j\}$ is phase retrievable if and only if $\{Sf_j\}$ is phase retrievable: the two lifted kernels are related by invertible matrices.  In particular, frame normalization preserves phase retrieval and turns the $f_jf_j^*$ into a POVM.

A frame has the \emph{complement property} if, for every $\Omega\subseteq\{1,\ldots,N\}$, either $\spanop\{f_j:j\in\Omega\}=\mathcal H$ or $\spanop\{f_j:j\notin\Omega\}=\mathcal H$.  Over $\R$, this property characterizes phase retrieval and the minimum length in $\R^n$ is $2n-1$.  Over $\C$, it is only necessary \cite{BalanCasazzaEdidin2006}.  Write
\begin{equation}
 d_{\mathbb F}(n)=\min\{N:\text{a phase-retrievable frame of $N$ vectors
 exists in $\mathbb F^n$}\}.
 \label{eq:d-field}
\end{equation}
For $n\ge2$, generic complex frames of length $4n-4$ are phase retrievable \cite{ConcaEtAl2015}, but $4n-4$ is not the exact minimum in every dimension. The low-dimensional values include $d_{\C}(2)=4$ and $d_{\C}(3)=8$ \cite{ConcaEtAl2015,HeinosaariMazzarellaWolf2013}. Combining Vinzant's eleven-vector construction with Huang and Kim's lower bound gives $d_{\C}(4)=11$ \cite{Vinzant2015,Huang2026}. These counts concern rank-one intensity measurements. For general positive effects, the corrected construction in \cite{HeinosaariMazzarellaWolf2013} uses ten self-adjoint observables, corresponding to a POVM with eleven outcomes in dimension four. The cited bounds leave the minimum number of outcomes of a general pure-state informationally complete POVM in $\C^4$ at either ten or eleven \cite{HeinosaariMazzarellaWolf2013,Huang2026}. These bounds do not establish whether general positive effects can use fewer outcomes than rank-one effects in this dimension.

\begin{example}[The same three vectors over two fields] \label{ex:three-vectors}
The family $\{e_1,e_2,e_1+e_2\}$ is phase retrievable in $\R^2$ by the complement property.  The same vectors fail in $\C^2$: the unit vectors $x=2^{-1/2}(1,\mathrm i)^t, \; y=2^{-1/2}(1,-\mathrm i)^t$ give the same three intensities, although $xx^*\ne yy^*$.  The three measurements match the diagonal and the real part of the coherence but miss its imaginary part.  Adding $e_1-\mathrm i e_2$ supplies that missing quantity and gives a four-vector phase-retrievable frame in $\C^2$.
\end{example}

\subsection{Positive operator measurements and POVMs}

A finite family $\{T_j\}_{j=1}^N\subset\B(\mathcal H)$ is an \emph{operator-valued frame} when
\begin{equation*}
 A I\le \sum_{j=1}^N T_j^*T_j\le B I
\end{equation*}
for some $A,B>0$; it is Parseval when the sum equals $I$.  This is the finite-dimensional form of the general operator-valued-frame definition \cite{KaftalLarsonZhang2009}.  Following \cite{HanJuste2019}, such a family is \emph{phase retrievable from quadratic measurements} if the values $\{\inner{x}{T_jx}\}_{j=1}^N$ determine $xx^*$.  For positive semidefinite measurement operators, there is an equivalent norm-measurement formulation: if $M_j\ge0$ and $B_j=M_j^{1/2}$, then
\begin{equation*}
 \inner{x}{M_jx}=\norm{B_jx}^2.
\end{equation*}
Thus $M_j$ is the measured effect, whereas $B_j$ is its operator-valued-frame factor.  Keeping these objects distinct avoids a minor but common notational ambiguity.

A self-adjoint family $\{M_j\}_{j=1}^N$ is \emph{pure-state informationally complete} when its expectation values distinguish all elements of $\Pone(\mathcal H)$ \cite{HeinosaariMazzarellaWolf2013}. This concerns unit vectors and need not imply phase retrieval for arbitrary vectors. A POVM is a positive family $F_j\ge0$ with $\sum_jF_j=I$. For a POVM, pure-state informational completeness is equivalent to phase retrieval from quadratic measurements, since the sum of the measurements determines $\norm{x}^2$. Its square roots form a Parseval operator-valued frame because
\begin{equation}
 \sum_j\norm{F_j^{1/2}x}^2=\norm{x}^2,
 \qquad
 \norm{F_j^{1/2}x}^2=\Tr(xx^*F_j).
 \label{eq:povm-factorization}
\end{equation}
For rank-one effects, this is exactly the classical correspondence between normalized tight frames and rank-one generalized quantum measurements \cite{EldarForney2002}.

The distinction between pure-state and full-state tomography is a distinction between kernel sets.  If $\mathcal L=\spanop_{\R}\{F_1,\ldots,F_N\}\subseteq\Herm(\mathcal H)$, then
\begin{equation}
 \begin{array}{ll}
 \text{all states are distinguished}
 &\Longleftrightarrow\ \mathcal L=\Herm(\mathcal H),\\[1mm]
 \text{all pure states are distinguished}
 &\Longleftrightarrow\ \mathcal L^\perp\cap\Sec(\mathcal H)=\{0\}.
 \end{array}
 \label{eq:full-versus-pure}
\end{equation}
Indeed, if $0\ne X\in\mathcal L^\perp$, then $\Tr X=0$ because $I\in\mathcal L$, and sufficiently small perturbations $I/n\pm tX$ are two states with the same data.  The pure-state assertion is Theorem~\ref{thm:lifted-kernel}.  Consequently full informational completeness implies pure-state informational completeness, but not conversely; see \cite{HeinosaariMazzarellaWolf2013,Finkelstein2004,FlammiaEtAl2005,CarmeliEtAl2014}. Terminology in the earlier literature is not uniform: some minimal-outcome results concern identification of almost every pure state, whereas ``pure-state informationally complete'' in this survey always means identification of every pure state.

\begin{lemma}[Cross-term criterion {\cite[Lemma~2.3]{HanJuste2019}}]
\label{lem:cross-term}
Let $\{T_j\}_{j=1}^N$ be an operator-valued frame on $\mathcal H$.  It is phase retrievable from quadratic measurements if and only if, whenever $u,v$ are nonzero and $u\notin\mathrm i\R v$, at least one $j$ satisfies
\begin{equation}
 \inner{u}{T_jv}+\overline{\inner{u}{T_j^*v}}\ne0.
 \label{eq:cross-term}
\end{equation}
For a self-adjoint family on a real Hilbert space this reduces to $\spanop\{T_jv:j\}=\mathcal H$ for every $v\ne0$.
\end{lemma}

\begin{proof}[Proof sketch]
For $a=u+v$ and $b=u-v$, direct expansion in the convention that the inner product is linear in its first variable gives
\begin{equation*}
 \inner{a}{T_ja}-\inner{b}{T_jb}
 =2\left(\inner{u}{T_jv}
 +\overline{\inner{u}{T_j^*v}}\right).
\end{equation*}
Moreover, $aa^*\ne bb^*$ precisely in the nondegenerate case represented by nonzero $u,v$ with $u\notin\mathrm i\R v$.  This proves the criterion. If the space is real and every $T_j$ is self-adjoint, the displayed cross term is $2\inner{u}{T_jv}$; its nonvanishing for every nonzero $u$ is equivalent to $\spanop\{T_jv:j\}=\mathcal H$.
\end{proof}

The criterion isolates the missing cross term: phase retrieval fails exactly when all measurements erase the same coherence.

\begin{example}[The tetrahedral qubit POVM]\label{ex:tetrahedral-povm}
Let $s_1,\ldots,s_4\in\R^3$ be the vertices of a regular tetrahedron on the unit sphere.  Thus $\sum_js_j=0$ and $\sum_js_js_j^t=\frac43I_3$.  With $\boldsymbol\sigma=(\sigma_x,\sigma_y,\sigma_z)$, define
\begin{equation*}
 F_j=\frac14(I+s_j\cdot\boldsymbol\sigma).
\end{equation*}
Each $F_j$ is positive and rank one, and $\sum_jF_j=I$.  If $\rho=\frac12(I+r\cdot\boldsymbol\sigma)$, then
\begin{equation}
 p_j=\Tr(\rho F_j)=\frac14(1+s_j\cdot r),
 \qquad r=3\sum_{j=1}^4p_js_j.
 \label{eq:tetrahedral-reconstruction}
\end{equation}
Hence the probabilities determine every qubit state, not only the pure ones.  This symmetric informationally complete POVM is the normalized four vector phase-retrieval model in quantum form \cite{RenesEtAl2004}.
\end{example}

\subsection{Quantum channels and pulled-back measurements}

A quantum channel $\Phi:\B(\mathcal H_A)\to\B(\mathcal H_B)$ is a completely positive and trace preserving map.  It has a Kraus representation
\begin{equation}
 \Phi(X)=\sum_{i=1}^rA_iXA_i^*,
 \qquad \sum_{i=1}^rA_i^*A_i=I_{\mathcal H_A},
 \label{eq:kraus-form}
\end{equation}
where the smallest possible $r$ is the Choi rank $\Cr(\Phi)$ \cite{Watrous2018,Choi1975,Kraus1971}.  The adjoint $\Phi^*(Y)=\sum_iA_i^*YA_i$ is unital.  Consequently an output POVM $\{F_j\}$ pulls back to the input POVM $\{\Phi^*(F_j)\}$, and setting
\begin{equation*}
 B_j=\Phi^*(F_j)^{1/2}
\end{equation*}
gives a Parseval operator-valued frame since $\sum_jB_j^*B_j=I$.  This verifies the operator-valued-frame structure behind \eqref{eq:information-flow}.

The following terminology combines \cite[Definition~1.1]{LiuHan2025} and \cite[Definition~2.1]{HanLiu2024}.

\begin{definition}
The channel $\Phi$ is \emph{phase retrievable} on a subspace $\mathcal M\subseteq\mathcal H_A$ if some output POVM $\{F_j\}$ has the property that $\{P_{\mathcal M}\Phi^*(F_j)P_{\mathcal M}\}$ distinguishes all pure states supported in $\mathcal M$.  The phase-retrievability index is
\begin{equation}
 \pr(\Phi)=\max\{\dim\mathcal M:\Phi\text{ is phase retrievable on }
 \mathcal M\}.
 \label{eq:pr-index}
\end{equation}
\end{definition}

The following theorem shows that phase retrievability of a channel on $\mathcal M$ is equivalent to its injectivity on the pure states supported in $\mathcal M$. It follows directly from \cite[Proposition~1.4]{LiuHan2025} and \cite[Lemma~5.1]{HanLiu2024}.

\begin{theorem}[Phase retrieval and pure-state injectivity]\label{thm:channel-injectivity}
A channel $\Phi$ is phase retrievable on $\mathcal M$ if and only if
\begin{equation}
 \Phi(xx^*)=\Phi(yy^*),\quad x,y\in\mathcal M
 \quad\Longrightarrow\quad xx^*=yy^*.
 \label{eq:pure-state-injectivity}
\end{equation}
\end{theorem}

\begin{proof}[Proof sketch]
If a pulled-back POVM separates pure inputs, equality of the two outputs gives equality of every pulled-back probability and hence of the inputs. Conversely, choose an informationally complete POVM on $\mathcal H_B$ \cite{HeinosaariMazzarellaWolf2013}. Equality of all its probabilities implies equality of the output density operators, after which \eqref{eq:pure-state-injectivity} applies.
\end{proof}

The output measurement in the converse is chosen to be sufficiently informative.  Pure-state injectivity does not imply that every fixed output POVM identifies the input; a coarse POVM may discard distinctions that remain present in the output operator.

The Choi-rank-one observation in \cite[Section~2]{LiuHan2025} is the basic positive case of the criterion.

Three intrinsic spaces help prevent another common conflation.  For a minimal Kraus family in \eqref{eq:kraus-form}, set
\begin{equation}
 \mathcal O_\Phi=\ran\Phi^*,
 \qquad
 \mathcal K_\Phi=\spanop\{A_i\},
 \qquad
 \mathcal S_\Phi=\spanop\{A_i^*A_j:i,j\}.
 \label{eq:three-spaces}
\end{equation}
Here the spans of Kraus operators are complex linear spans.  The Hermitian part of the observable range $\mathcal O_\Phi$ controls what the receiver can measure at the input.  Under the standard vectorization $A\mapsto|A\rangle\!\rangle$, the Kraus span $\mathcal K_\Phi$ corresponds to the support of the Choi matrix and carries the distinct quadratic-measurement problem studied for square Kraus operators in \cite{HanJuste2019}.  The noise operator system $\mathcal S_\Phi$ controls zero-error confusability and quantum correction \cite{DuanSeveriniWinter2013,KnillLaflamme1997}.  The isometric freedom in minimal Kraus representations shows that all three spaces are representation independent \cite{Watrous2018}, but they need not coincide. Table~\ref{tab:three-spaces} summarizes their distinct roles.

\begin{table}[ht]
\centering
\small
\renewcommand{\arraystretch}{1.15}
\begin{tabular}{
>{\raggedright\arraybackslash}p{0.20\textwidth}
>{\raggedright\arraybackslash}p{0.34\textwidth}
>{\raggedright\arraybackslash}p{0.34\textwidth}}
\toprule
\textbf{Space}&\textbf{Frame/operator role}&\textbf{Operational task}\\
\midrule
$\mathcal O_\Phi$&pulled-back observables&pure-state and full-state identification\\
$\mathcal K_\Phi$&Choi support; Kraus quadratic measurements&representation-side phase retrieval\\
$\mathcal S_\Phi$&noise/confusability operator system&zero-error discrimination and correction\\
\bottomrule
\end{tabular}
\caption{Three channel spaces that should not be interchanged.}
\label{tab:three-spaces}
\end{table}

Finally, for $W_x:\C^r\to\mathcal H_B$ given by $W_xc=\sum_ic_iA_ix$,
\begin{equation}
 \Phi(xx^*)=W_xW_x^*=\sum_i(A_ix)(A_ix)^*,
 \qquad
 W_x^*W_x=[\inner{A_jx}{A_ix}]_{i,j}.
 \label{eq:kraus-image-frame}
\end{equation}
The output is therefore the frame operator of the output-vector family $\{A_ix\}_{i=1}^r$.  A complementary channel associated with the chosen Kraus representation is
\begin{equation}
 \Phi^c(X)=\bigl[\Tr(A_iXA_j^*)\bigr]_{i,j=1}^r,
 \qquad
 \Phi^c(xx^*)=\bigl[\inner{A_ix}{A_jx}\bigr]_{i,j=1}^r,
 \label{eq:complementary-channel}
\end{equation}
so its pure-state output is the transpose of $W_x^*W_x$, with the entries given above \cite{Watrous2018}. In particular,
\begin{equation}
 \Tr\bigl(\Phi(xx^*)^2\bigr)
 =\sum_{i,j=1}^r\abs{\inner{A_ix}{A_jx}}^2,
 \label{eq:output-frame-potential}
\end{equation}
is the frame potential of $\{A_ix\}$ \cite{BenedettoFickus2003}.  Thus output rank, output purity, and overlaps between two outputs are respectively span, frame potential, and cross-Gram data of the corresponding output vector families.  This is a second, concrete reason that frame theory is natural for channel analysis.

%%%%%%%%%%%%%%%%%%%%%%%%%%%%%%%%%%%%%%%%%%%%%%%%%%%%%%%%%%%%%%%%%%%%%%%%%%%%%%%%%%%%%%%%%%%%%%%%%%%%%%%%%%%%%%%%%
%%%%%%%%%%%%%%%%%%%%%%%%%%%%%%%%%%%%%%%%%%%%%%%%%%%%%%%%%%%%%%%%%%%%%%%%%%%%%%%%%%%%%%%%%%%%%%%%%%%%%%%%%%%%%%%%%
%%%%%%%%%%%%%%%%%%%%%%%%%%%%%%%%%%%%%%%%%%%%%%%%%%%%%%%%%%%%%%%%%%%%%%%%%%%%%%%%%%%%%%%%%%%%%%%%%%%%%%%%%%%%%%%%%
%%%%%%%%%%%%%%%%%%%%%%%%%%%%%%%%%%%%%%%%%%%%%%%%%%%%%%%%%%%%%%%%%%%%%%%%%%%%%%%%%%%%%%%%%%%%%%%%%%%%%%%%%%%%%%%%%
%%%%%%%%%%%%%%%%%%%%%%%%%%%%%%%%%%%%%%%%%%%%%%%%%%%%%%%%%%%%%%%%%%%%%%%%%%%%%%%%%%%%%%%%%%%%%%%%%%%%%%%%%%%%%%%%%

\section{Pure-state criteria and frame-guided channel design}
\label{sec:criteria-design}

\subsection{Secants and stable separation}

Theorem~\ref{thm:channel-injectivity} turns phase retrievability into an intrinsic property of the channel.  The qualitative criteria below combine \cite[Proposition~1.4 and Lemma~2.7]{LiuHan2025}; the stability clause is the standard compact normalized secant consequence.  This gives a field independent test and a modest stability statement without introducing a separate recovery theory.

\begin{theorem}[Secant criterion and stability]
\label{thm:channel-secant-stability}
Let $\mathcal M\subseteq\mathcal H_A$ have dimension at least two and define
\begin{equation}
 \eta_2(\Phi;\mathcal M)=
 \min_{\substack{X\in\Sec(\mathcal M)\\\norm{X}_2=1}}
 \norm{\Phi(X)}_2.
 \label{eq:stability-modulus}
\end{equation}
For $\mathcal M=\mathcal H_A$ we abbreviate $\eta_2(\Phi;\mathcal H_A)$ to $\eta_2(\Phi)$. The following are equivalent:
\begin{enumerate}[(i)]
\item $\Phi$ is phase retrievable on $\mathcal M$;
\item $\ker\Phi\cap\Sec(\mathcal M)=\{0\}$;
\item $\bigl(\Phi^*(\Herm(\mathcal H_B))\bigr)^\perp \cap\Sec(\mathcal M)=\{0\}$, where the orthogonal complement is taken in the real Hilbert space $\Herm(\mathcal H_A)$;
\item $\eta_2(\Phi;\mathcal M)>0$.
\end{enumerate}
When these conditions hold, all unit $x,y\in\mathcal M$ satisfy
\begin{equation}
 \eta_2(\Phi;\mathcal M)\norm{xx^*-yy^*}_2
 \le \norm{\Phi(xx^*)-\Phi(yy^*)}_2
 \le \norm{\Phi}_{2\to2}\norm{xx^*-yy^*}_2.
 \label{eq:channel-stability}
\end{equation}
\end{theorem}

\begin{proof}[Proof sketch]
The first two conditions are equivalent by Theorems~\ref{thm:lifted-kernel} and~\ref{thm:channel-injectivity}; the equivalence with (iii) follows from the Hilbert--Schmidt identity $\bigl(\Phi^*(\Herm(\mathcal H_B))\bigr)^\perp=\ker\Phi\cap\Herm(\mathcal H_A)$.  The normalized secant set in \eqref{eq:stability-modulus} is compact.  Hence the continuous function $X\mapsto\norm{\Phi(X)}_2$ has a positive minimum exactly when the kernel intersection is trivial.  Applying the definition to the normalized difference of two pure states proves the lower estimate, and the induced operator norm gives the upper estimate.  This is the channel version of the standard injectivity stability principle for finite-dimensional phase retrieval \cite{BandeiraEtAl2014,BalanWang2015,GrohsKoppensteinerRathmair2020}. Recent condition-number analysis for intensity measurements develops this quantitative viewpoint further \cite{PengHanHuang2026}.
\end{proof}

\begin{corollary}[Observable-range and post-processing monotonicity]\label{cor:postprocessing-monotonicity} Let
$\Phi:\B(\mathcal H_A)\to\B(\mathcal H_B)$ and $\Psi:\B(\mathcal H_A)\to\B(\mathcal H_C)$ be channels.  If
\begin{equation}
 \Phi^*(\Herm(\mathcal H_B))
 \subseteq \Psi^*(\Herm(\mathcal H_C)),
 \label{eq:observable-range-inclusion}
\end{equation}
then every subspace that is phase retrievable for $\Phi$ is phase retrievable for $\Psi$, and hence $\pr(\Phi)\le\pr(\Psi)$.  In particular, for every channel $\Theta:\B(\mathcal H_B)\to\B(\mathcal H_C)$,
\begin{equation}
 \pr(\Theta\circ\Phi)\le\pr(\Phi).
 \label{eq:postprocessing-monotonicity}
\end{equation}
\end{corollary}

\begin{proof}[Proof sketch]
Taking orthogonal complements reverses the inclusion in \eqref{eq:observable-range-inclusion}; the secant criterion in Theorem~\ref{thm:channel-secant-stability} then proves the first claim. For the second, $(\Theta\circ\Phi)^*=\Phi^*\circ\Theta^*$, so the observable range of the composite is contained in that of $\Phi$.
\end{proof}

If an informationally complete output POVM $\{F_j\}$ is fixed, its analysis map has a positive lower singular value on $\Herm(\mathcal H_B)$.  Combining that bound with \eqref{eq:channel-stability} shows that robust experimental data require two ingredients: the channel must separate pure-state secants, and the chosen output POVM must read the surviving directions with an adequate measurement bound.  The existence of an identifying POVM alone does not guarantee a well-conditioned implementation.

\begin{example}[Information-preserving and information-erasing extremes]\label{ex:extreme-channels}
Assume $\dim\mathcal H_A\ge2$. If $V^*V=I$ and $\Phi_V(X)=VXV^*$, then $\Phi_V$ is injective on every operator and preserves Hilbert--Schmidt and trace norms.  Thus it is phase retrievable and $\eta_2(\Phi_V)=1$.  At the opposite extremal case, a replacement channel $\mathcal R_\sigma(X)=\Tr(X)\sigma$ maps every pure state to the same density operator and has $\pr(\mathcal R_\sigma)=1$ when $\dim\mathcal H_A>1$.  These are the Choi-rank-one isometric case from \cite{LiuHan2025} and a simple example of complete loss of state information.
\end{example}

\begin{example}[Depolarization and dephasing]\label{ex:depolarizing-dephasing}
For $n\ge2$ and $0<p\le1$, the depolarizing channel on $M_n$ is
\begin{equation*}
 \Delta_p(X)=pX+(1-p)\Tr(X)\frac{I_n}{n}.
\end{equation*}
Every pure-state secant is traceless, so $\Delta_p(xx^*-yy^*)=p(xx^*-yy^*)$.  Hence $\Delta_p$ is phase retrievable and $\eta_2(\Delta_p)=p$.  In contrast, complete dephasing on a qubit,
\begin{equation}
 \mathcal Z(X)=e_1e_1^*Xe_1e_1^*+e_2e_2^*Xe_2e_2^*,
 \label{eq:qubit-dephasing}
\end{equation}
maps the two pure states associated with $(e_1+e_2)/\sqrt2$ and $(e_1-e_2)/\sqrt2$ to the same diagonal state. It therefore has $\eta_2(\mathcal Z)=0$ and $\pr(\mathcal Z)=1$.  The missing information is exactly the relative phase between the two basis components.
\end{example}

The tetrahedral POVM from Example~\ref{ex:tetrahedral-povm} gives a useful intermediate case.  Upon outcome $j$, prepare the pure state $\rho_j=\frac12(I+s_j\cdot\boldsymbol\sigma)$.  The measure-and-prepare channel
\begin{equation*}
 \Phi_{\rm tet}(\rho)=\sum_{j=1}^4\Tr(\rho F_j)\rho_j
\end{equation*}
satisfies, by \eqref{eq:tetrahedral-reconstruction},
\begin{equation}
 \Phi_{\rm tet}\!\left(\frac12(I+r\cdot\boldsymbol\sigma)\right)
 =\frac12\left(I+\frac13r\cdot\boldsymbol\sigma\right).
 \label{eq:tetrahedral-channel}
\end{equation}
It is entanglement breaking and injective on the entire qubit state space. Thus preservation of enough classical statistics for identification does not require preservation of entanglement \cite{HorodeckiShorRuskai2003}.

\subsection{Choi rank two and linear Kraus-operator combinations}

Suppose
\begin{equation}
 \Phi(X)=A_1XA_1^*+A_2XA_2^*,
 \qquad A_1^*A_1+A_2^*A_2=I.
 \label{eq:choi-rank-two}
\end{equation}
When this representation is minimal, $\Cr(\Phi)=2$.  The following criterion applies to any two-operator Kraus representation, in the minimal case, it is the sharp characterization of Choi-rank-two channels established in \cite[Theorem~2.2]{LiuHan2025}.

\begin{theorem}[Two-Kraus operators criterion]\label{thm:two-kraus-combinations}
The channel in \eqref{eq:choi-rank-two} is phase retrievable if and only if, for every $\lambda\in\C$, at least one of
\begin{equation}
 A_1+\lambda A_2,
 \qquad
 -\overline\lambda A_1+A_2
 \label{eq:two-kraus-combinations}
\end{equation}
is injective.
\end{theorem}

\begin{proof}[Proof sketch]
If both operator combinations have nonzero kernel vectors $x$ and $y$, their two kernel relations cancel the cross term
\begin{equation*}
 \inner{x}{\Phi^*(F)y}
 =\inner{A_1x}{FA_1y}+\inner{A_2x}{FA_2y}
\end{equation*}
for every output observable $F$.  The pointwise spanning consequence of Lemma~\ref{lem:cross-term} then rules out phase retrieval.  Conversely, the condition at $\lambda=0$ makes one Kraus operator injective. If two inequivalent pure inputs had the same output, comparison of the two rank-two decompositions produces a $2\times2$ coefficient matrix whose distinct eigenvalues $\alpha_1,\alpha_2$ satisfy $\alpha_1\overline{\alpha_2}=-1$.  Corresponding eigenvectors then give a value of $\lambda$ for which both combinations in \eqref{eq:two-kraus-combinations} have a kernel, contradiction.  The coefficient lemma and full calculation are given in \cite[Lemma~2.1 and Theorem~2.2]{LiuHan2025}.
\end{proof}

The coefficient vectors $(1,\lambda)$ and $(-\overline\lambda,1)$ are orthogonal in $\C^2$.  Thus the theorem says that two orthogonal combinations of a minimal two-element Kraus family cannot simultaneously lose injectivity.

\begin{example}[Amplitude damping]\label{ex:amplitude-damping}
For $0\le\gamma\le1$, let
\begin{equation*}
 A_1=\begin{pmatrix}1&0\\0&\sqrt{1-\gamma}\end{pmatrix},
 \qquad
 A_2=\sqrt\gamma\begin{pmatrix}0&1\\0&0\end{pmatrix}.
\end{equation*}
The trace-preserving identity is immediate.  If $\gamma<1$, then
\begin{equation*}
 \det(A_1+\lambda A_2)=\sqrt{1-\gamma}\ne0
 \qquad(\lambda\in\C),
\end{equation*}
so Theorem~\ref{thm:two-kraus-combinations} proves phase retrievability. When $\gamma=1$, both basis states are sent to $e_1e_1^*$, the channel then becomes a reset channel and is not phase retrievable.
\end{example}

\subsection{The higher-rank relative-spectrum criterion}

For arbitrary Choi rank, scalar coefficient relations extend to matrix-valued relative spectra. The cancellation mechanism underlying this condition is independent of the scalar field. We first state its field independent operator form from \cite[Corollary~2.6]{LiuHan2025}, and then explain its necessary phase-retrieval consequence over either field and its sufficiency for complex channels of Choi rank two \cite[Theorem~2.2]{LiuHan2025}.

Let $\mathbf C=(C_1,\ldots,C_p)^t$ and $\mathbf D=(D_1,\ldots,D_q)^t$ be operator tuples from $H$ to $K$.  In finite dimensions, define
\begin{equation}
 \sigma_\ell(\mathbf C,\mathbf D)=
 \left\{\Lambda\in M_{p\times q}(\mathbb F):
 \begin{array}{l}
 \text{there is }0\ne z\in H\text{ such that}\\[-1mm]
 C_i z=\sum_{j=1}^q\lambda_{ij}D_jz,\quad1\le i\le p
 \end{array}\right\},
 \label{eq:relative-spectrum-definition}
\end{equation}
where $\mathbb F=\R$ or $\C$.  Equivalently, the block-column operator $\mathbf C-\Lambda\mathbf D:H\to K^p$ is not injective, and also not left invertible in finite dimensions. For an ordered subset $\Omega\subset\{1,\ldots,r\}$, write $\mathbf A_\Omega=(A_i)_{i\in\Omega}^t$.

\begin{theorem}[Field-independent relative spectrum criterion]\label{thm:relative-spectrum}
Let $\Phi(X)=\sum_{i=1}^rA_iXA_i^*$ be a channel over $\mathbb F\in\{\R,\C\}$ of Choi rank $r$, with minimal Kraus family $A_1,\ldots,A_r$.  The following are equivalent:
\begin{enumerate}[(i)]
\item $\Phi(x\rtensor y)\ne0$ for all nonzero $x,y$;
\item for every ordered nontrivial bipartition
$\Omega\sqcup\Omega^c=\{1,\ldots,r\}$ and every coefficient matrix
$\Lambda\in M_{|\Omega|\times|\Omega^c|}(\mathbb F)$, at least one of the following conditions holds:
\begin{equation}
 \Lambda\notin\sigma_\ell(\mathbf A_\Omega,\mathbf A_{\Omega^c}),
 \qquad
 -\Lambda^*\notin\sigma_\ell(\mathbf A_{\Omega^c},\mathbf A_\Omega).
 \label{eq:relative-spectrum-condition}
\end{equation}
\end{enumerate}
\end{theorem}

\begin{proof}[Proof sketch]
If both conditions in \eqref{eq:relative-spectrum-condition} fail, then there are nonzero $x,y$ such that
\begin{equation}
 \mathbf A_\Omega x=\Lambda\mathbf A_{\Omega^c}x,
 \qquad
 \mathbf A_{\Omega^c}y=-\Lambda^*\mathbf A_\Omega y.
 \label{eq:relative-spectrum-relations}
\end{equation}
Substitution into $\sum_iA_ix\rtensor A_iy$ gives pairwise cancellation, so $\Phi(x\rtensor y)=0$.

Conversely, suppose $\Phi(x\rtensor y)=0$ with $x,y\ne0$.  The family $\{A_ix\}$ is linearly dependent, otherwise a dual family applied to $0=\sum_iA_iy\rtensor A_ix$ would imply $A_iy=0$ for every $i$, contrary to $\sum_iA_i^*A_i=I$.  Choose $\Omega^c$ so that $\{A_jx:j\in\Omega^c\}$ is a basis of $\spanop\{A_ix\}$.  Expressing the remaining $A_ix$ in that basis produces $\Lambda$ and the first relation in \eqref{eq:relative-spectrum-relations}.  Linear independence in the identity $\sum_iA_ix\rtensor A_iy=0$ gives the second relation.  Hence both conditions in \eqref{eq:relative-spectrum-condition} fail.  This is exactly the cancellation argument of \cite[Corollary~2.6]{LiuHan2025}.
\end{proof}

\begin{corollary}[Phase retrieval and the relative-spectrum condition] \label{cor:relative-spectrum-phase-retrieval} Let $\Phi$ be the channel defined in Theorem~\ref{thm:relative-spectrum}.
\begin{enumerate}[(i)]
\item Over either $\R$ or $\C$, phase retrievability implies the equivalent conditions in Theorem~\ref{thm:relative-spectrum}.
\item Over $\C$, if $\Cr(\Phi)=2$, these conditions are also sufficient for phase retrievability.
\end{enumerate}
\end{corollary}

\begin{proof}[Proof sketch]
Theorem~\ref{thm:relative-spectrum} converts the relative spectrum condition into $\Phi(x\rtensor y)\ne0$ for all nonzero $x,y$. If $\Phi(x\rtensor y)=0$, trace preservation gives $\inner{x}{y}=0$. Since $\Phi$ preserves adjoints, the nonzero self-adjoint operator $x\rtensor y+y\rtensor x$ belongs to $\ker\Phi\cap\Sec(\mathcal H_A)$, contradicting phase retrievability by Theorem~\ref{thm:channel-secant-stability}. This proves necessity over either field. When the complex Choi rank is two, the relative spectrum condition reduces to the paired linear combination criterion in Theorem~\ref{thm:two-kraus-combinations}, which proves sufficiency.
\end{proof}
The relative spectrum condition is necessary for phase retrievability over either field. Over the complex field it is also sufficient at Choi rank two, whereas at higher Choi rank the condition is generally only necessary. The following construction, adapted from \cite[Example~2.1]{LiuHan2025}, illustrates the complex higher-rank limitation.
\begin{example}[Complex higher-rank limitation]\label{ex:complex-higher-rank-gap}
In $\C^2$, take $f_1=e_1$, $f_2=e_2$, and $f_3=e_1+e_2$.  Their frame has the complement property but is not phase retrievable, as seen in Example~\ref{ex:three-vectors}.  Put
\begin{equation*}
 R_j=f_jf_j^*,\qquad S=\sum_{j=1}^3R_j^2,\qquad
 B_j=R_jS^{-1/2},
\end{equation*}
and define $\Psi(X)=\sum_{j=1}^3B_jXB_j^*$.  Since $\sum_jB_j^*B_j=I$, this is a channel.  Linear independence of the $R_j$ and the complement property imply that $\Psi(x\rtensor y)\ne0$ for all nonzero $x,y$.  Nevertheless, the two inequivalent vectors with equal frame intensities give, after the invertible change by $S^{1/2}$, two distinct vectors with equal outputs.  Trace preservation makes their norms equal, so division by their common norm gives two normalized pure inputs with the same output.  Hence $\Psi$ satisfies the relative-spectrum condition but is not phase retrievable.  This is the frame-based construction of \cite[Example~2.1]{LiuHan2025}.
\end{example}

The following negative result is complementary: instead of detecting a spectral cancellation, it identifies a structural loss of all cross-summand coherences \cite[Proposition~3.5]{LiuHan2025}.

\begin{proposition}[Direct-sum Kraus-range obstruction]\label{prop:direct-sum-obstruction}
Let $\Phi(X)=\sum_{i=1}^rA_iXA_i^*$ have the minimal Kraus family $A_1,\ldots,A_r$, where $r=\Cr(\Phi)>1$.  If
\begin{equation*}
 \ran A_1^*+\cdots+\ran A_r^*
\end{equation*}
is a direct sum, then $\Phi$ is not phase retrievable.  In particular, a nontrivial pinching channel $X\mapsto\sum_iP_iXP_i$ with mutually orthogonal nonzero projections is not phase retrievable.
\end{proposition}

\begin{proof}[Proof sketch]
For each $i$, choose a basis $\{g_{ij}\}_{j=1}^{k_i}$ of $\ran A_i^*$ and write $A_i^*=\sum_jg_{ij}\rtensor f_{ij}$.  Directness makes the union of the $g_{ij}$ linearly independent. Extend it to a basis of the input space and let $h_{ij}$ denote the corresponding dual vectors.  Minimality makes two ranges nonzero, say the first two.  Then $x=h_{11}+h_{21}$ and $y=h_{11}-h_{21}$ have different rank-one operators, while direct calculation gives
\begin{equation*}
 \Phi(xx^*)=f_{11}f_{11}^*+f_{21}f_{21}^*=\Phi(yy^*).
\end{equation*}
For orthogonal projections, this is precisely the loss of relative phase, which is already seen in \eqref{eq:qubit-dephasing}.
\end{proof}

There is also a useful positive higher-rank class \cite[Corollary~3.4]{LiuHan2025}.  It shows that the presence of several rank-one Kraus terms need not destroy phase retrievability when their geometry is tied to an invertible positive term.

\begin{proposition}\label{prop:positive-kraus-class}
Consider a channel on $\B(H)$ with positive Kraus operators $A_1,\ldots,A_r$.  If $A_r$ is invertible and $A_1,\ldots,A_{r-1}$ have rank one, then $\Phi$ is phase retrievable.
\end{proposition}

\begin{proof}[Proof sketch]
Write $A_j=u_ju_j^*$ for $j<r$ and put $f_j=A_r^{-1}u_j$.  For $F_j=f_jf_j^*$, the first $r-1$ pulled-back effects are linear combinations of $u_ku_k^*$ whose coefficient matrix is
\begin{equation*}
 M=I_{r-1}+
 \left[\abs{\inner{A_r^{-1/2}u_k}{A_r^{-1/2}u_j}}^2
 \right]_{k,j=1}^{r-1}.
\end{equation*}
The second term is positive semidefinite by the Schur product theorem \cite[Theorem~7.5.3]{HornJohnson2013}, so $M$ is invertible.  Extend $u_1,\ldots,u_{r-1}$ to a phase-retrievable frame $u_1,\ldots,u_N$ for $H$, and for $j\ge r$ put $f_j=A_r^{-1}u_j$ and $F_j=f_jf_j^*$.  Equality of the first $r-1$ pulled-back data lets one invert $M$ and recover the corresponding frame intensities.  Subtracting their known contributions from the remaining data recovers $\abs{\inner{x}{u_j}}^2$ for every $j\ge r$.  The extended frame is phase retrievable, so the output POVM obtained by scaling these positive tests and adding a residual effect distinguishes all pure inputs.
\end{proof}

\subsection{Channels designed from phase-retrievable frames}

The converse design problem begins with rank-one output measurement operators and asks for a channel whose pullbacks form a phase-retrievable family. The following construction is the representative prescribed-Choi-rank result of \cite[Proposition~3.6]{LiuHan2025}, with trace preservation made explicit.

\begin{theorem}[Prescribed Choi rank from rank-one POVMs]\label{thm:prescribed-choi-rank}
Assume $\dim\mathcal H_A=\dim\mathcal H_B$.  Let $F_j=f_jf_j^*$, $1\le j\le N$, be linearly independent rank-one positive operators of a POVM such that $\{f_j\}_{j=1}^N$ is phase retrievable and
\begin{equation}
 I\notin\spanop\{F_j:j\in\Lambda\}
 \quad\text{for every proper }\Lambda\subsetneq\{1,\ldots,N\}.
 \label{eq:no-proper-identity}
\end{equation}
For every $1\le r\le N$, there is a channel $\Phi_r$ of Choi rank $r$ such that $\{\Phi_r^*(F_j)\}_{j=1}^N$ distinguishes all pure input states.
\end{theorem}

\begin{proof}
Choose a unitary $U:\mathcal H_A\to\mathcal H_B$, set $g_j=U^*f_j$, and let $K_j=f_j\rtensor g_j=F_jU$.  For $r\ge2$ define the completely positive map
\begin{equation*}
 \Psi_r(X)=UXU^*+\sum_{j=1}^{r-1}K_jXK_j^*,
\end{equation*}
with only the unitary term when $r=1$.  Condition \eqref{eq:no-proper-identity} and linear independence of the $F_j$ imply that $U,K_1,\ldots,K_{r-1}$ are linearly independent, hence $\Cr(\Psi_r)=r$.

For the tests $F_k$, the pulled-back quadratic data are
\begin{equation}
 \inner{x}{\Psi_r^*(F_k)x}
 =|\inner{x}{g_k}|^2+
 \sum_{j=1}^{r-1}|\inner{f_k}{f_j}|^2|\inner{x}{g_j}|^2.
 \label{eq:intensity-mixing}
\end{equation}
For $1\le k<r$, the coefficient matrix is $I+[|\inner{f_k}{f_j}|^2]_{k,j=1}^{r-1}$, which is positive definite by the Schur product theorem \cite[Theorem~7.5.3]{HornJohnson2013}.  The first $r-1$ data therefore recover the corresponding $|\inner{x}{g_j}|^2$, subtraction in the remaining equations recovers every frame intensity.  Since $\{g_j\}$ is phase retrievable, $\{\Psi_r^*(F_j)\}$ separates pure states.

It remains to normalize the map.  Let $S=\Psi_r^*(I)>0$ and set
\begin{equation}
 \Phi_r(X)=\Psi_r(S^{-1/2}XS^{-1/2}).
 \label{eq:channel-normalization}
\end{equation}
Its Kraus operators are the preceding Kraus operators multiplied on the right by $S^{-1/2}$, so their squared products sum to $I$ and their linear independence is preserved.  Moreover, $\Phi_r^*(F_j)=S^{-1/2}\Psi_r^*(F_j)S^{-1/2}$, which still preserves pure-state separation.  Thus $\Phi_r$ is trace preserving, has Choi rank $r$, and has the required pullback measurements.
\end{proof}

The following is the minimal-frame construction from \cite[Lemma~3.1, Proposition~3.6 and Corollary~3.7]{LiuHan2025}, the same proof also gives the stated complex extension.  For $\mathbb F=\R$, ``channel'', ``adjoint'', and ``POVM'' refer to their real matrix analogues, and the standard quantum interpretation usually refers to the complex case.

\begin{corollary}[Minimum Parseval construction]\label{cor:minimum-parseval}
Over $\mathbb F=\R$ or $\C$, let $\{f_j\}_{j=1}^{d_{\mathbb F}(n)}$ be a minimum length phase-retrievable Parseval frame for $\mathbb F^n$.  For every $1\le r\le d_{\mathbb F}(n)$ there is a channel $\Phi_r$ of Choi rank $r$ such that
\begin{equation*}
 \{\Phi_r^*(f_jf_j^*)\}_{j=1}^{d_{\mathbb F}(n)}
\end{equation*}
is a phase-retrievable POVM.
\end{corollary}

\begin{proof}[Proof sketch]
Minimality implies that the effects $f_jf_j^*$ are linearly independent, otherwise one measurement could be deleted without changing the kernel in Theorem~\ref{thm:lifted-kernel}.  Parsevality gives $\sum_jf_jf_j^*=I$, so linear independence also gives \eqref{eq:no-proper-identity}.  Apply Theorem~\ref{thm:prescribed-choi-rank}. Since $\Phi_r^*$ is unital, the pulled-back effects still sum to $I$.
\end{proof}

The preceding constructions encode frame data in the pullback of an output measurement.  A complementary construction records the measurement probabilities directly in linearly independent prepared states.

\begin{proposition}[Measure-and-prepare factorization]\label{prop:measure-prepare-factorization}
Let $\{F_j\}_{j=1}^N$ be a POVM on $\mathcal H_A$, let $\sigma_1,\ldots,\sigma_N\in\D(\mathcal H_B)$ be real linearly independent, and define
\begin{equation}
 \Phi(X)=\sum_{j=1}^N\Tr(F_jX)\sigma_j.
 \label{eq:general-measure-prepare}
\end{equation}
For every subspace $\mathcal M\subseteq\mathcal H_A$, the channel $\Phi$ is phase retrievable on $\mathcal M$ if and only if the compressed POVM $\{P_{\mathcal M}F_jP_{\mathcal M}\}_{j=1}^N$ is pure-state informationally complete on $\mathcal M$.  Moreover, $\Phi$ is injective on all density operators if and only if $\{F_j\}$ is fully informationally complete.
\end{proposition}

\begin{proof}[Proof sketch]
The map in \eqref{eq:general-measure-prepare} is completely positive because its Choi matrix is a sum of positive tensor products, and it is trace preserving because $\Tr\sigma_j=1$ and $\sum_jF_j=I$.  For Hermitian $X$ and $Y$, real linear independence of the $\sigma_j$ gives
\begin{equation*}
 \Phi(X)=\Phi(Y)
 \quad\Longleftrightarrow\quad
 \Tr(F_jX)=\Tr(F_jY)\quad(1\le j\le N).
\end{equation*}
Restricting first to pure states supported in $\mathcal M$, and then to all density operators, proves the two assertions. Channels of the form \eqref{eq:general-measure-prepare} are entanglement breaking \cite{HorodeckiShorRuskai2003}.
\end{proof}

\begin{example}[Pure-state information without full-state information]\label{ex:pure-not-full}
Let $\{f_j\}_{j=1}^8$ be a phase-retrievable Parseval frame in $\C^3$ and put $F_j=f_jf_j^*$.  Choose eight real linearly independent density operators $\sigma_j\in M_3$ and define
\begin{equation}
 \Phi(X)=\sum_{j=1}^8\Tr(F_jX)\sigma_j.
 \label{eq:qutrit-measure-prepare}
\end{equation}
Proposition~\ref{prop:measure-prepare-factorization} shows that this is an entanglement breaking, phase-retrievable channel.  But the measurement map has only eight real coordinates on the nine-dimensional space $\Herm(\C^3)$, and hence has a nonzero traceless kernel element.  Small positive and negative perturbations of $I/3$ give two mixed states with the same output.  This channel makes the pure-state/full-state distinction in \eqref{eq:full-versus-pure} concrete and reflects the phenomenon in \cite[Example~1.1]{LiuHan2025}.
\end{example}

%%%%%%%%%%%%%%%%%%%%%%%%%%%%%%%%%%%%%%%%%%%%%%%%%%%%%%%%%%%%%%%%%%%%%%%%%%%%%%%%%%%%%%%%%%%%%%%%%%%%%%%%%%%%%%%%%
%%%%%%%%%%%%%%%%%%%%%%%%%%%%%%%%%%%%%%%%%%%%%%%%%%%%%%%%%%%%%%%%%%%%%%%%%%%%%%%%%%%%%%%%%%%%%%%%%%%%%%%%%%%%%%%%%
%%%%%%%%%%%%%%%%%%%%%%%%%%%%%%%%%%%%%%%%%%%%%%%%%%%%%%%%%%%%%%%%%%%%%%%%%%%%%%%%%%%%%%%%%%%%%%%%%%%%%%%%%%%%%%%%%
%%%%%%%%%%%%%%%%%%%%%%%%%%%%%%%%%%%%%%%%%%%%%%%%%%%%%%%%%%%%%%%%%%%%%%%%%%%%%%%%%%%%%%%%%%%%%%%%%%%%%%%%%%%%%%%%%
%%%%%%%%%%%%%%%%%%%%%%%%%%%%%%%%%%%%%%%%%%%%%%%%%%%%%%%%%%%%%%%%%%%%%%%%%%%%%%%%%%%%%%%%%%%%%%%%%%%%%%%%%%%%%%%%%

\section{Identification, perfect discrimination, and correction}
\label{sec:distinction}

Phase retrievability says that distinct pure inputs have distinct output density operators.  It does not say that one use of the channel permits a perfect decision between them.  Two density operators $\rho$ and $\sigma$ are perfectly distinguishable in one shot precisely when their supports are orthogonal, equivalently $\Tr(\rho\sigma)=0$ \cite{Helstrom1976,Watrous2018}.  Nonorthogonal pure states, for example, remain distinct under the identity channel but cannot be perfectly distinguished in one shot.

Trace distance quantifies this gap.  For two output states occurring with equal prior probabilities, the optimal one-shot success probability is \cite{Helstrom1976}
\begin{equation}
 p_{\rm succ}(\rho,\sigma)
 =\frac12+\frac14\norm{\rho-\sigma}_1.
 \label{eq:helstrom-success}
\end{equation}
Injectivity is equivalent, pair by pair, to this number being strictly greater than $1/2$, and perfect discrimination requires it to equal one. Consequently, a positive secant modulus controls how rapidly statistical distinguishability can deteriorate under perturbation, but it does not turn nonorthogonal outputs into a zero-error alphabet.

Following the noncommutative graph formulation in \cite{DuanSeveriniWinter2013} and the notation of \cite{HanLiu2024}, the \emph{independence number} of a channel $\Phi(X)=\sum_iA_iXA_i^*$ is
\begin{equation}
 \alpha(\Phi)=\max\left\{N:
 \begin{array}{l}
 \text{there are orthonormal }x_1,\ldots,x_N\text{ such that}\\[-1mm]
 \supp\Phi(x_rx_r^*)\perp\supp\Phi(x_sx_s^*)\quad(r\ne s)
 \end{array}\right\}.
 \label{eq:independence-number}
\end{equation}
Expanding the output overlap gives
\begin{equation}
 \Tr\bigl(\Phi(xx^*)\Phi(yy^*)\bigr)
 =\sum_{i,j}|\inner{x}{A_i^*A_jy}|^2.
 \label{eq:output-overlap}
\end{equation}
Thus perfect distinguishability is controlled by the noise operator system $\mathcal S_\Phi$ in \eqref{eq:three-spaces}, rather than by $\mathcal O_\Phi$.  Regularization yields the zero-error classical capacity
\begin{equation}
 C_0(\Phi)=\lim_{q\to\infty}\frac1q\log_2\alpha(\Phi^{\otimes q}).
 \label{eq:zero-error-capacity}
\end{equation}
This is the noncommutative graph viewpoint of \cite{DuanSeveriniWinter2013}.

A subspace $C\subseteq\mathcal H_A$ is an \emph{exactly correctable code} if a recovery channel $\mathcal R$ satisfies $\mathcal R\Phi(\rho)=\rho$ for every state supported in $C$.  Writing $P_C$ as its projection, the Knill--Laflamme criterion \cite{KnillLaflamme1997} is
\begin{equation}
 P_CA_i^*A_jP_C=\lambda_{ij}P_C
 \quad(i,j),
 \qquad [\lambda_{ij}]\ge0.
 \label{eq:knill-laflamme}
\end{equation}
Let $\beta(\Phi)=\max\{\dim C:C\text{ is exactly correctable for }\Phi\}.$ Every orthonormal basis of a correctable code gives a zero-error alphabet, so $\beta(\Phi)\le\alpha(\Phi)$ \cite{HanLiu2024}.

The next implication is standard, compare \cite[Lemmas~3.1 and~5.2]{HanLiu2024} for its use with twirling channels.

\begin{theorem}[Correctability implies phase retrievability]\label{thm:correctability-bridge}
If $C$ is exactly correctable for $\Phi$, then
\begin{equation}
 \norm{\Phi(X)}_1=\norm{X}_1
 \qquad\text{for every Hermitian }X=P_CXP_C.
 \label{eq:trace-isometry-on-code}
\end{equation}
Consequently $C$ is phase retrievable for $\Phi$ and
\begin{equation}
 \beta(\Phi)\le\min\{\alpha(\Phi),\pr(\Phi)\}.
 \label{eq:beta-bounds}
\end{equation}
\end{theorem}

\begin{proof}[Proof sketch]
Exact recovery extends by real linearity from states to every Hermitian operator in $P_C\B(\mathcal H_A)P_C$.  Trace-norm contractivity of channels on Hermitian operators \cite[Chapter~3]{Watrous2018} gives
\begin{equation*}
 \norm{X}_1=\norm{\mathcal R\Phi(X)}_1
 \le\norm{\Phi(X)}_1\le\norm{X}_1,
\end{equation*}
which proves \eqref{eq:trace-isometry-on-code}.  Equality of two pure-state outputs in $C$ then forces equality of their inputs.  Combining this with the zero-error consequence of \eqref{eq:knill-laflamme} gives \eqref{eq:beta-bounds}.
\end{proof}

\begin{corollary}[Stability on an exactly correctable code]\label{cor:code-stability}
Let $C$ be exactly correctable for $\Phi:\B(\mathcal H_A)\to\B(\mathcal H_B)$.  Then, for all density operators $\rho,\sigma$ supported in $C$,
\begin{equation}
 \norm{\Phi(\rho)-\Phi(\sigma)}_1=\norm{\rho-\sigma}_1.
 \label{eq:code-trace-distance}
\end{equation}
Thus the optimal success probability is preserved on $C$.  If $\dim C\ge2$ and $d_B=\dim\mathcal H_B$, then
\begin{equation}
 \eta_2(\Phi;C)\ge d_B^{-1/2}.
 \label{eq:code-eta-lower-bound}
\end{equation}
\end{corollary}

\begin{proof}[Proof sketch]
Apply \eqref{eq:trace-isometry-on-code} to $\rho-\sigma$ and then use \eqref{eq:helstrom-success}.  If $X\in\Sec(C)$ and $\norm{X}_2=1$, then
\begin{equation*}
 \norm{\Phi(X)}_2\ge d_B^{-1/2}\norm{\Phi(X)}_1
 =d_B^{-1/2}\norm{X}_1\ge d_B^{-1/2}.
\end{equation*}
Taking the minimum over normalized secants proves \eqref{eq:code-eta-lower-bound}.
\end{proof}

Table~\ref{tab:three-operational-notions} compares the three notions.

\begin{table}[ht]
\centering
\small
\renewcommand{\arraystretch}{1.15}
\begin{tabular}{
>{\raggedright\arraybackslash}p{0.17\textwidth}
>{\raggedright\arraybackslash}p{0.38\textwidth}
>{\raggedright\arraybackslash}p{0.34\textwidth}}
\toprule
\textbf{Quantity}&\textbf{Question}&\textbf{Controlling condition}\\
\midrule
$\pr(\Phi)$&How large a subspace has all its pure states injectively identified?&$\ker\Phi$ avoids pure-state secants\\
$\alpha(\Phi)$&How many messages are perfectly distinguished in one use?&pairwise $\mathcal S_\Phi$-orthogonality\\
$\beta(\Phi)$&How large a quantum subspace is exactly recovered?&$P_C\mathcal S_\Phi P_C=\C P_C$\\
\bottomrule
\end{tabular}
\caption{Three operational notions and their controlling conditions.}
\label{tab:three-operational-notions}
\end{table}

\begin{example}[Injective and zero-error distinguishability are incomparable]\label{ex:incomparability}
For $0<p<1$, every output of $\Delta_p$ in Example~\ref{ex:depolarizing-dephasing} is positive definite.  Hence $\alpha(\Delta_p)=\beta(\Delta_p)=1$, although $\pr(\Delta_p)=n$.  Thus this single family already separates pure-state identification from exact quantum correction.  Conversely, qubit dephasing preserves the two computational basis states as perfectly distinguishable outputs, so $\alpha(\mathcal Z)=2$, while $\pr(\mathcal Z)=1$.  Thus neither $\alpha$ nor $\pr$ bounds the other. The current valid relation is the code bound $\beta\le\min\{\alpha,\pr\}$ from Theorem~\ref{thm:correctability-bridge}.
\end{example}

%%%%%%%%%%%%%%%%%%%%%%%%%%%%%%%%%%%%%%%%%%%%%%%%%%%%%%%%%%%%%%%%%%%%%%%%%%%%%%%%%%%%%%%%%%%%%%%%%%%%%%%%%%%%%%%%%
%%%%%%%%%%%%%%%%%%%%%%%%%%%%%%%%%%%%%%%%%%%%%%%%%%%%%%%%%%%%%%%%%%%%%%%%%%%%%%%%%%%%%%%%%%%%%%%%%%%%%%%%%%%%%%%%%
%%%%%%%%%%%%%%%%%%%%%%%%%%%%%%%%%%%%%%%%%%%%%%%%%%%%%%%%%%%%%%%%%%%%%%%%%%%%%%%%%%%%%%%%%%%%%%%%%%%%%%%%%%%%%%%%%
%%%%%%%%%%%%%%%%%%%%%%%%%%%%%%%%%%%%%%%%%%%%%%%%%%%%%%%%%%%%%%%%%%%%%%%%%%%%%%%%%%%%%%%%%%%%%%%%%%%%%%%%%%%%%%%%%
%%%%%%%%%%%%%%%%%%%%%%%%%%%%%%%%%%%%%%%%%%%%%%%%%%%%%%%%%%%%%%%%%%%%%%%%%%%%%%%%%%%%%%%%%%%%%%%%%%%%%%%%%%%%%%%%%

\section{Twirling channels and phase-retrievable subspaces}
\label{sec:twirling}

Twirling channels form a structured setting in which the observable range, the noise algebra, and the relevant frame geometry are all computable from a representation.  Let $G$ be a compact group with normalized Haar measure $\mu$, and let $\pi:G\to\mathcal U(H)$ be a finite-dimensional unitary representation.  Its twirling channel is
\begin{equation}
 \Phi_\pi(T)=\int_G\pi(g)T\pi(g)^*\,d\mu(g).
 \label{eq:twirling-channel}
\end{equation}
This is the standard symmetry-averaging channel and the trace-preserving conditional expectation onto the representation commutant \cite{VollbrechtWerner2001,GirardLevick2021}
\begin{equation*}
 \pi(G)'=\{T\in\B(H):T\pi(g)=\pi(g)T\text{ for every }g\in G\}.
\end{equation*}

Choose an irreducible decomposition
\begin{equation}
 H=\bigoplus_{i=1}^d(\C^{m_i}\otimes H_i),
 \qquad
 \pi=\bigoplus_{i=1}^d(I_{m_i}\otimes\pi_i),
 \qquad n_i=\dim H_i,
 \label{eq:irreducible-decomposition}
\end{equation}
where the $\pi_i$ are pairwise inequivalent irreducible representations. Here, $m_i$ is the multiplicity and $n_i$ the irreducible dimension. The tensor $\C^{m_i}\otimes H_i$ is the $i$th \emph{isotypic component}. We use the standard complete reducibility and Schur orthogonality theory for compact groups \cite{Simon1996}.

\subsection{Schur averaging and the commutant}

Let $\mathcal A_\pi$ be the $C^*$-algebra generated by $\pi(G)$.  Then
\begin{equation}
 \mathcal A_\pi=
 \bigoplus_{i=1}^d(I_{m_i}\otimes\B(H_i)),
 \qquad
 \mathcal A_\pi'=
 \bigoplus_{i=1}^d(M_{m_i}(\C)\otimes I_{H_i}).
 \label{eq:twirling-algebras}
\end{equation}
The twirl is the trace-preserving conditional expectation, equivalently the Hilbert--Schmidt orthogonal projection, onto $\mathcal A_\pi'$. Let $P_i$ denote the orthogonal projection from $H$ onto the isotypic component $\C^{m_i}\otimes H_i$. If $P_iTP_i=[T_{ab}^{(i)}]_{a,b=1}^{m_i}$ and $\Tr_{H_i}$ denotes the partial trace over the irreducible factor $H_i$, Schur orthogonality gives the standard block formula \cite[Section~2]{HanLiu2024}
\begin{equation}
 \Phi_\pi(T)=
 \bigoplus_{i=1}^d
 \left[\frac{\Tr(T_{ab}^{(i)})}{n_i}I_{H_i}\right]_{a,b=1}^{m_i}
 =\bigoplus_{i=1}^d
 \left(\frac1{n_i}\Tr_{H_i}(P_iTP_i)\otimes I_{H_i}\right).
 \label{eq:schur-block-formula}
\end{equation}
Thus off-diagonal blocks between inequivalent sectors vanish. Within an isotypic component, the irreducible factor is depolarized, while the full matrix structure on the multiplicity factor survives.

\begin{corollary}[Global phase retrievability of a twirl]\label{cor:global-twirl}
The twirling channel $\Phi_\pi$ is phase retrievable on all of $H$ if and only if $d=1$ and $n_1=1$.  Equivalently, $\pi(g)=\chi(g)I_H$ for a unitary character $\chi$ and $\Phi_\pi$ is the identity channel.
\end{corollary}

\begin{proof}
If $d\ge2$, choose unit vectors $u,v$ in two inequivalent isotypic components.  The off-diagonal blocks vanish, so the distinct pure states of $(u+v)/\sqrt2$ and $(u-v)/\sqrt2$ have the same twirled output.  If $d=1$ but $n_1>1$, choose a unit multiplicity vector $a\in\C^{m_1}$ and orthogonal unit vectors $u,v\in H_1$.  Formula \eqref{eq:schur-block-formula} sends both $(a\otimes u)(a\otimes u)^*$ and $(a\otimes v)(a\otimes v)^*$ to $a a^*\otimes I_{H_1}/n_1$.  Thus global phase retrievability fails in both cases.  If $d=1$ and $n_1=1$, the representation acts by a scalar character and conjugation is the identity.
\end{proof}

For a twirl, the three spaces from \eqref{eq:three-spaces} simplify to
\begin{equation}
 \mathcal O_{\Phi_\pi}=\mathcal A_\pi',
 \qquad
 \mathcal K_{\Phi_\pi}=\mathcal A_\pi,
 \qquad
 \mathcal S_{\Phi_\pi}=\mathcal A_\pi.
 \label{eq:twirling-three-spaces}
\end{equation}
The first identity follows from conditional expectation.  For the other two, the vectorization $A\mapsto|A\rangle\!\rangle$ identifies the support of the Choi matrix with the span of $\pi(G)$, and the products $\pi(g)^*\pi(h)$ span the same algebra.  Hence symmetry makes both the adjoint range and the noise operator system explicit: they are mutual commutants, which need not coincide \cite{Watrous2018,GirardLevick2021}.

An elementary consequence connects the channel kernel to orbit-frame orthogonality:
\begin{equation}
 \Phi_\pi(x\rtensor y)=0
 \quad\Longleftrightarrow\quad
 x\perp\mathcal A_\pi' y.
 \label{eq:commutant-kernel}
\end{equation}
This is \cite[Lemma~4.1]{HanLiu2024} in commutant form. Indeed, $\Phi_\pi=\Phi_\pi^*$ and its range is $\mathcal A_\pi'$.  For finite $G$, the left side is also equivalent to vanishing of the cross-frame operator between the two orbit families $\{\pi(g)x\}_{g\in G}$ and $\{\pi(g)y\}_{g\in G}$. Each nonzero orbit family is a frame for its cyclic span. In frame terminology, the two orbit frames are strongly disjoint precisely when this cross-frame operator vanishes \cite[Lemma~4.1]{HanLiu2024}\cite{HanLarson2000}. Recent work on phase retrieval by irreducible and projective representation frames provides complementary representation side discussions \cite{Cheng2025,ChengLi2026}.

This interpretation also explains the twirl directly in frame language.
For a finite group,
\begin{equation}
 \Phi_\pi(xx^*)=\frac1{|G|}\sum_{g\in G}
 (\pi(g)x)(\pi(g)x)^*,
 \label{eq:orbit-frame-operator}
\end{equation}
the output is the normalized frame operator of the orbit of $x$, regarded as a frame for its cyclic span.  The channel therefore identifies two pure inputs exactly when their orbit frames have different frame operators.  This is stronger than asking the two orbits to be different as sets and weaker than retaining the ordered orbit vectors.  Equation~\eqref{eq:commutant-kernel} identifies the lost cross-frame operators, while the multiplicity-free result below converts the surviving diagonal block weights back into scalar frame measurements.

\begin{example}[Irreducible twirls and qubit dephasing]
\label{ex:twirling-extremes}
If $\pi$ is irreducible on an $n$-dimensional space, then
\begin{equation*}
 \Phi_\pi(T)=\frac{\Tr(T)}nI,
\end{equation*}
the twirl is completely depolarizing and $\pr(\Phi_\pi)=1$ for $n>1$. For the other extremal case, let $G=\mathbb Z_2$ and $\pi(1)=Z=\diag(1,-1)$ on $\C^2$.  Then
\begin{equation*}
 \Phi_\pi(T)=\frac12(T+ZTZ)=\diag(T),
\end{equation*}
which is precisely the dephasing channel in Example~\ref{ex:depolarizing-dephasing}.  These two cases already show how irreducibility and separation into inequivalent characters can both erase the coherences needed for global pure-state identification.
\end{example}

\subsection{Coding and orthogonality indices}

The work in \cite{HanLiu2024} also introduces two off-diagonal quantities. The orthogonality index $\gamma(\Phi)$ is the largest $N$ for which there are nonzero $x_1,\ldots,x_N$ satisfying $\Phi(x_r\rtensor x_s)=0$ for $r\ne s$.  Its subspace version is
\begin{equation*}
 \tau(\Phi)=\max\{\dim M:\Phi(x\rtensor y)=0
 \text{ whenever }x,y\in M,\ x\perp y\}.
\end{equation*}
The corresponding exact formulas for twirling channels are collected from \cite[Theorems~3.1--3.3, 4.1 and~4.2]{HanLiu2024} in the following compact statement.

\begin{theorem}[Five indices of a twirling channel]
\label{thm:five-indices}
For the decomposition \eqref{eq:irreducible-decomposition},
\begin{equation}
 \begin{gathered}
 \alpha(\Phi_\pi)=\sum_{i=1}^d m_i,
 \qquad \beta(\Phi_\pi)=\max_{1\le i\le d}m_i,
 \qquad C_0(\Phi_\pi)=\log_2\!\left(\sum_{i=1}^d m_i\right),\\[-1mm]
 \gamma(\Phi_\pi)=\sum_{i=1}^d n_i,
 \qquad \tau(\Phi_\pi)=\max_{1\le i\le d}n_i.
 \end{gathered}
 \label{eq:twirling-five-indices}
\end{equation}
In particular, $\alpha(\Phi_\pi)=\beta(\Phi_\pi)$ exactly when $d=1$.
\end{theorem}

\begin{proof}[Proof sketch]
The noise algebra is $\mathcal A_\pi$.  Fixing one unit vector in every irreducible factor gives $\sum_im_i$ mutually $\mathcal A_\pi$-orthogonal vectors, one from each copy.  Conversely, the cyclic invariant subspaces generated by a zero-error family are nonzero and mutually orthogonal, so their number cannot exceed the number of irreducible summands counted with multiplicity.  This proves the formula for $\alpha$.  For a fixed unit $u\in H_i$, the subspace $\C^{m_i}\otimes u$ is correctable.  Projecting any correctable code into a nonzero isotypic block and applying the Knill--Laflamme relations gives the reverse bound $\dim C\le\max_i m_i$.

For $q$ independent uses, $\Phi_\pi^{\otimes q}$ is the twirl of the external tensor product representation of $G^q$.  Its multiplicities are $m_{i_1}\cdots m_{i_q}$, and hence $\alpha(\Phi_\pi^{\otimes q})=(\sum_i m_i)^q$, which proves the capacity identity in \eqref{eq:twirling-five-indices}.  On one block $\pi=m\sigma$, relation \eqref{eq:commutant-kernel} reduces to mutual orthogonality of the irreducible factor coefficient spaces and yields $\gamma(\Phi_{m\sigma})=\dim H_\sigma$.  Combining the blocks gives $\gamma=\sum_in_i$. The corresponding subspace argument also gives $\tau=\max_in_i$.  See the cited theorems for the complete arguments.
\end{proof}

The formulas display the dual role of multiplicities and irreducible dimensions.  Large $m_i$ record degrees of freedom that survive the twirl and can support messages or a correctable code subspace, while large $n_i$ measure how many directions can have their cross terms erased.

For the $\mathbb Z_2$ example, $d=2$ and all $m_i=n_i=1$, hence $\alpha=\gamma=2$ while $\beta=\tau=1$.  Together with $\pr=1$, this is the smallest example of a channel supporting two perfectly distinguishable classical messages without identifying an arbitrary pure qubit.

\subsection{Phase-retrievable subspaces}

The block formula also turns the phase-retrieval question into a positive operator measurement on a parameter space.  The cleanest exact statement is the multiplicity-free case.  The next result is equivalent to the discussion in  \cite[Proposition~5.2]{HanLiu2024}.

\begin{theorem}[Multiplicity-free compression]
\label{thm:multiplicity-free}
Assume $m_1=\cdots=m_d=1$, so that $\pi=\pi_1\oplus\cdots\oplus\pi_d$.  Let $V:\C^k\to H$ be an isometry with range $M$, put $V_i=P_iV$, and define
\begin{equation}
 E_i=V_i^*V_i\ge0,
 \qquad 1\le i\le d.
 \label{eq:compressed-effects}
\end{equation}
Then $\sum_iE_i=I_k$, and $M$ is phase retrievable for $\Phi_\pi$ if and only if the POVM $\{E_i\}_{i=1}^d$ is pure-state informationally complete on $\C^k$.
\end{theorem}

\begin{proof}[Proof sketch]
Since $V$ is an isometry and $\sum_iP_i=I_H$, $\sum_iE_i=V^*V=I_k$.  Formula \eqref{eq:schur-block-formula} gives, for $\xi\in\C^k$,
\begin{equation}
 \Phi_\pi(V\xi\xi^*V^*)
 =\bigoplus_{i=1}^d
 \frac{\inner{\xi}{E_i\xi}}{n_i}I_{H_i}.
 \label{eq:multiplicity-free-output}
\end{equation}
Thus two pure states in $M$ have the same output exactly when their expectations against every $E_i$ agree.  Theorem \ref{thm:channel-injectivity} proves the equivalence. 
\end{proof}

\begin{corollary}[Weighted stability formula]
\label{cor:multiplicity-free-stability}
Under the hypotheses and notation of Theorem~\ref{thm:multiplicity-free}, assume in addition
that $k=\dim M\ge2$. Then
\begin{equation}
 \eta_2(\Phi_\pi;M)=
 \min_{\substack{X\in\Sec(\C^k)\\\norm{X}_2=1}}
 \left(\sum_{i=1}^d
 \frac{\abs{\Tr(E_iX)}^2}{n_i}\right)^{1/2}.
 \label{eq:multiplicity-free-stability}
\end{equation}

\end{corollary}

\begin{proof}[Proof sketch]
The isometry $V$ preserves Hilbert--Schmidt norms and carries $\Sec(\C^k)$ onto $\Sec(M)$.  By linearity of
\eqref{eq:multiplicity-free-output}, every Hermitian $X$ satisfies
\begin{equation*}
 \Phi_\pi(VXV^*)=
 \bigoplus_{i=1}^d\frac{\Tr(E_iX)}{n_i}I_{H_i}.
\end{equation*}
The squared Hilbert--Schmidt norm of the right side is $\sum_i\abs{\Tr(E_iX)}^2/n_i$.  Taking the minimum over normalized secants proves the formula.  
\end{proof}

The theorem shows where ordinary frame phase retrieval re-enters.  In the multiplicity-free case, the surviving scalar block norms are the outcomes of a POVM on the coordinates of the candidate subspace.  For nontrivial multiplicities, the partial traces in \eqref{eq:schur-block-formula} retain matrix-valued rather than merely scalar data. Thus, in general, the problem does not reduce to phase retrieval from $d$ scalar-valued measurements.

\begin{example}[A two-dimensional phase-retrievable subspace]\label{ex:four-characters}
Let $G=\mathbb Z_4$ and let $\pi=\chi_1\oplus\cdots\oplus\chi_4$ be the sum of its four distinct characters.  Start with the injective map
\begin{equation*}
 T:\C^2\to\C^4,
 \qquad T(a,b)=(a,b,a+b,a+\mathrm i b),
\end{equation*}
and normalize it to the isometry $V=T(T^*T)^{-1/2}$.  The four coordinate effects of $T$ have quadratic data
\begin{equation*}
 |a|^2,\quad |b|^2,\quad |a+b|^2,\quad |a+\mathrm i b|^2.
\end{equation*}
The third determines $\operatorname{Re}(a\overline b)$ and the fourth determines $\operatorname{Im}(a\overline b)$ once $|a|^2,|b|^2$ are known.  They therefore recover $\left(\begin{smallmatrix}|a|^2&a\overline b\\\overline a b&|b|^2\end{smallmatrix}\right)$. Invertible normalization preserves phase retrieval, so Theorem~\ref{thm:multiplicity-free} makes $\ran V$ a two-dimensional phase-retrievable subspace.  Here $\beta=1$ but $\pr\ge2$, visibly separating correctability from pure-state identification.
\end{example}

The lower bound and the exact character case below combine \cite[Lemma~5.2, Corollary~5.1, Theorem~5.1 and the discussion following Theorem~5.1]{HanLiu2024}.

\begin{theorem}[Lower bound for the phase-retrievability index]\label{thm:pr-lower-bound}
For \eqref{eq:irreducible-decomposition},
\begin{equation}
 \pr(\Phi_\pi)\ge
 \max\left\{\max_i m_i,\ \left\lfloor\frac d4\right\rfloor+1\right\}.
 \label{eq:pr-lower-bound}
\end{equation}
If $\pi=\chi_1\oplus\cdots\oplus\chi_d$ is a sum of distinct characters and
\begin{equation}
 d_\C(k)\le d<d_\C(k+1),
 \label{eq:character-window}
\end{equation}
then $\pr(\Phi_\pi)=k$.
\end{theorem}

\begin{proof}[Proof sketch]
The subspace $\C^{m_i}\otimes u$, for a fixed unit $u\in H_i$, retains its full multiplicity matrix under the partial trace and is phase retrievable. This gives $\pr\ge\max_i m_i$, which also follows from $\beta\le\pr$.  Select one copy of each irreducible sector.  For $k=\lfloor d/4\rfloor+1\ge2$, the inequality $d\ge4k-4$ allows a phase-retrievable frame of $d$ vectors to be placed into those sectors, and Theorem~\ref{thm:multiplicity-free} yields a $k$-dimensional subspace. The case $k=1$ is immediate from any one-dimensional subspace.

In the character case every compressed effect has rank at most one. If a phase-retrievable subspace had dimension $L\ge k+1$, its restriction to any $(k+1)$-dimensional subspace would yield a phase-retrievable vector frame of at most $d$ nonzero vectors in $\C^{k+1}$. This would imply $d_\C(k+1)\le d$, contrary to \eqref{eq:character-window}. Hence $\pr(\Phi_\pi)\le k$.

Conversely, since $d_\C(k)\le d$, choose a phase-retrievable frame in $\C^k$ with at most $d$ vectors, normalize it to a Parseval frame, and append zero vectors if necessary to obtain $f_1,\ldots,f_d$. Choose a unit vector $u_i$ in each character sector and define
\begin{equation*}
 V\xi=\bigoplus_{i=1}^d\inner{\xi}{f_i}u_i.
\end{equation*}
Parsevality makes $V$ an isometry, and its compressed effects are $f_if_i^*$. Theorem~\ref{thm:multiplicity-free} therefore makes $\ran V$ a $k$-dimensional phase-retrievable subspace. Thus $\pr(\Phi_\pi)\ge k$, proving equality.
\end{proof}

The newer exact value $d_\C(4)=11$ \cite{Huang2026} therefore sharpens the character case: a direct sum of $d$ distinct characters has $\pr=3$ for $8\le d\le10$, while $d=11$ is the first length at which the exact formula permits a four-dimensional phase-retrievable subspace.  Since $d_\C(5)=16$, the same formula gives $\pr=4$ throughout $11\le d\le15$ \cite{ConcaEtAl2015,HeinosaariMazzarellaWolf2013}.

%%%%%%%%%%%%%%%%%%%%%%%%%%%%%%%%%%%%%%%%%%%%%%%%%%%%%%%%%%%%%%%%%%%%%%%%%%%%%%%%%%%%%%%%%%%%%%%%%%%%%%%%%%%%%%%%%
%%%%%%%%%%%%%%%%%%%%%%%%%%%%%%%%%%%%%%%%%%%%%%%%%%%%%%%%%%%%%%%%%%%%%%%%%%%%%%%%%%%%%%%%%%%%%%%%%%%%%%%%%%%%%%%%%
%%%%%%%%%%%%%%%%%%%%%%%%%%%%%%%%%%%%%%%%%%%%%%%%%%%%%%%%%%%%%%%%%%%%%%%%%%%%%%%%%%%%%%%%%%%%%%%%%%%%%%%%%%%%%%%%%
%%%%%%%%%%%%%%%%%%%%%%%%%%%%%%%%%%%%%%%%%%%%%%%%%%%%%%%%%%%%%%%%%%%%%%%%%%%%%%%%%%%%%%%%%%%%%%%%%%%%%%%%%%%%%%%%%
%%%%%%%%%%%%%%%%%%%%%%%%%%%%%%%%%%%%%%%%%%%%%%%%%%%%%%%%%%%%%%%%%%%%%%%%%%%%%%%%%%%%%%%%%%%%%%%%%%%%%%%%%%%%%%%%%

\section{Conclusion and outlook}
\label{sec:conclusion}

The rank-one lift gives a single language between phaseless frame data and pure quantum states.  In that language, a channel retains every pure input exactly when its kernel avoids the normalized rank-two secants, or, equivalently, when its adjoint range contains enough input observables to separate those secants.  Compactness then turns exact separation into the two-sided finite-dimensional estimate in Theorem~\ref{thm:channel-secant-stability}.

Kraus coordinates remain valuable when their field dependence is kept explicit. The matrix-valued relative spectrum characterizes the absence of nonzero rank-one operators from the channel kernel over both real and complex Hilbert spaces. This condition is necessary for phase retrievability in every Choi rank and is also sufficient over complex Hilbert spaces at Choi rank two, while Example~\ref{ex:complex-higher-rank-gap} exhibits the higher-rank limitation. The normalized pure-state secant criterion in Theorem~\ref{thm:channel-secant-stability} supplies the corresponding field-uniform formulation.  Conversely, phase-retrievable Parseval frames give physically normalized channel constructions with prescribed Choi rank.

The operational comparisons are equally sharp.  Pure-state identification is governed by the observable range, whereas perfect one-shot discrimination and exact correction are governed by orthogonality and scalar compression conditions for the noise operator system.  Correctability implies both identification and a zero-error alphabet on the code, but neither phase retrievability nor the one-shot independence number controls the other for arbitrary channels. For twirling channels, Schur averaging makes this separation transparent: multiplicities determine the zero-error and correction indices, irreducible dimensions determine the orthogonality indices, and multiplicity-free compressions reduce phase-retrievable subspaces to ordinary POVMs and frames.  Corollary~\ref{cor:multiplicity-free-stability} adds a quantitative version of this reduction: the channel modulus becomes a weighted stability modulus of the compressed measurement on normalized pure-state secants.

\subsection*{Related uses of frames in quantum information}

The connection developed in this survey reflects a broader interplay between frame theory and quantum information theory. Rank-one POVMs correspond, through their measurement vectors, to normalized tight frames, and the canonical frame construction has a least-squares quantum measurement interpretation \cite{EldarForney2002}.  Informationally complete measurements can more generally be organized as frames in Hilbert--Schmidt operator space, duals yield reconstruction formulas, while group covariance and tightness produce structured and statistically favorable measurements \cite{DArianoPerinottiSacchi2004,Scott2006}.  Operator-valued frames also provide dilation and representation tools for structured quantum channels \cite{HanLiMengTang2011}.  These works concern reconstruction and channel structure beyond the particular pure-state channel injectivity question of this work, but they rely on the same basic principles: redundancy, tightness, duality, and symmetry.

Several more recent developments make the connection computational and statistical.  Quantum-detection problems have been studied with continuous and multi-window Gabor frames \cite{HanHuLiu2021,HanHuLiuWang2022}. Adaptive informationally complete POVMs have been optimized to reduce measurement variance in variational quantum algorithms \cite{GarciaPerezEtAl2021}.  In shadow tomography, general measurement frames express classical shadow estimators through dual frames, and noncanonical duals can be optimized to improve observable estimation \cite{InnocentiEtAl2023,FischerEtAl2024}.  Recent work on inclusion-minimal rank-one POVMs studies measurements for which pure-state identification fails after any outcome is discarded and only the conditional distribution of the remaining outcomes is retained \cite{EdidinGonzalezTamo2025}. Here the remaining probabilities are renormalized to sum to one; merely omitting one absolute probability does not have this effect, since it is determined by the others.  On the channel side, coherent interferometric coupling of operator-valued frames introduces cross terms that can enhance phase retrievability \cite{LiuHanNour2026}.  The exact minimum of eleven rank-one intensity measurements in $\C^4$ discussed earlier further sharpens the measurement side boundary \cite{Huang2026}.

\subsection*{Open questions and future directions}

This wider perspective suggests several concrete problems.  One is to minimize the number of physically normalized output effects while maximizing the combined channel and measurement stability bound.  A second is to relate the normalized secant modulus directly to sample complexity, dual-frame variance, and shadow-estimation performance.  A third is to determine exact phase-retrievability indices and stability moduli for covariant channels with nontrivial multiplicity spaces, where the surviving data are matrix valued.  It is also natural to extend the channel criteria from pure states to bounded-rank state families while retaining practical POVM constructions.  Finally, it is natural to ask which additional algebraic or symmetry assumptions ensure that robust pure-state identification coexists with nontrivial zero-error alphabets or correctable codes. The examples in this chapter show that this implication fails for general channels, but it remains an interesting question for restricted classes of channels.

%%%%%%%%%%%%%%%%%%%%%%%%%%%%%%%%%%%%%%%%%%%%%%%%%%%%%%%%%%%%%%%%%%%%%%%%%%%%%%%%%%%%%%%%%%%%%%%%%%%%%%%%%%%%%%%%%
%%%%%%%%%%%%%%%%%%%%%%%%%%%%%%%%%%%%%%%%%%%%%%%%%%%%%%%%%%%%%%%%%%%%%%%%%%%%%%%%%%%%%%%%%%%%%%%%%%%%%%%%%%%%%%%%%
%%%%%%%%%%%%%%%%%%%%%%%%%%%%%%%%%%%%%%%%%%%%%%%%%%%%%%%%%%%%%%%%%%%%%%%%%%%%%%%%%%%%%%%%%%%%%%%%%%%%%%%%%%%%%%%%%
%%%%%%%%%%%%%%%%%%%%%%%%%%%%%%%%%%%%%%%%%%%%%%%%%%%%%%%%%%%%%%%%%%%%%%%%%%%%%%%%%%%%%%%%%%%%%%%%%%%%%%%%%%%%%%%%%
%%%%%%%%%%%%%%%%%%%%%%%%%%%%%%%%%%%%%%%%%%%%%%%%%%%%%%%%%%%%%%%%%%%%%%%%%%%%%%%%%%%%%%%%%%%%%%%%%%%%%%%%%%%%%%%%%

\end{document}